\documentclass[12pt]{article}
\usepackage{graphicx}
\usepackage{natbib} %comment out if you do not have the package
\usepackage{url} % not crucial - just used below for the URL 
\usepackage{amssymb,amsmath,amsthm}

\newtheorem{proposition}{Proposition}

\newcommand{\Exp}[1]{{\rm{E}}[ #1 ]}
\newcommand{\Var}[1]{{\rm{Var}}[ #1 ]}

\newcommand{\tr}{\text{\rm trace}}

\newcommand{\bl}[1]{{\mathbf #1}}
\newcommand{\bs}[1]{\boldsymbol #1}

\newcommand{\blind}{0}

\begin{document}

\def\spacingset#1{\renewcommand{\baselinestretch}%
{#1}\small\normalsize} \spacingset{1}

%%%%%%%%%%%%%%%%%%%%%%%%%%%%%%%%%%%%%%%%%%%%%%%%%%%%%%%%%%%%%%%%%%%%%%%%%%%%%%

\if0\blind
{
  \title{\bf Linear Source Apportionment using Generalized Least Squares}
  \author{Jordan Bryan\thanks{
    The authors gratefully acknowledge \textit{Chris Osburn for providing the Neuse River Dataset}}\hspace{.2cm}\\
    and \\
    Peter Hoff \\
    Department of Statistical Science, Duke University}
  \maketitle
} \fi

\if1\blind
{
  \bigskip
  \bigskip
  \bigskip
  \begin{center}
    {\LARGE\bf Linear Source Apportionment using Generalized Least Squares}
\end{center}
  \medskip
} \fi

\bigskip
\begin{abstract}
Motivated by applications to water quality monitoring using fluorescence spectroscopy, we develop the source apportionment model for high dimensional profiles of dissolved organic matter (DOM). We describe simple methods to estimate the parameters of a linear source apportionment model, and show how the estimates are related to those of ordinary and generalized least squares. Using this least squares framework, we analyze the variability of the estimates, and we propose predictors for missing elements of a DOM profile. We demonstrate the practical utility of our results on fluorescence spectroscopy data collected from the Neuse River in North Carolina.
\end{abstract}

\noindent%
\vfill

\newpage
\spacingset{2} % DON'T change the spacing!
\section{Introduction}\label{sec:intro}

Increasing land development 
and the growth of large-scale agricultural operations
have led to concerns about water pollution 
and a need for 
quantitative methods for water quality monitoring. 
The water quality of a river basin is affected by the water quality 
of the streams that feed into it, which in turn are affected by 
the land-use features of their local watersheds. 
As a result, the water at a particular point of a river 
will contain a mixture of dissolved organic matter (DOM) 
whose sources are determined by the upstream land use. 
For example, the DOM profile of the water at a point 
downstream from both a poultry 
farm and a community septic system will resemble a mixture 
of the DOM profiles of water near the farm and of water near the septic system. 

In order to monitor pollution and the sources of DOM in the 
Neuse River basin in Eastern North Carolina, 
researchers at North Carolina State University obtained 
202 water samples, each one being representative of one of nine different categories of land use.
Fluorescence spectroscopy was used to obtain a multivariate DOM 
profile for each water sample. 
Taken together, these 202 profiles 
make up a ``dictionary'' 
to which the DOM profile of a water sample obtained downstream can be compared
\citep{osburn_predicting_2016}. 
In particular, it is of interest to estimate in what proportions 
each of the nine source categories contribute to the DOM profile 
of a downstream water sample. Such estimates can identify water quality issues
and provide information about land use and drainage patterns in the 
river basin. 

A DOM profile is often represented as a matrix, having elements 
that record the fluorescence intensity spectra emitted by a water sample when it is excited with light of a range of frequencies (see Figure \ref{fig:ex_eems}). 
In the water chemistry literature it is common practice to 
stack these ``excitation-emission matrices'' (EEMs) 
to form a three-way array, and then to analyze the data array using multiway statistical methods, in 
particular, 
the PARAFAC model. 
\citet{osburn_predicting_2016} developed a PARAFAC-based method 
called FluorMod to
estimate the 
source proportions of a downstream water sample from its DOM profile
and the dictionary of profiles from the nine source categories. 
While providing promising results, FluorMod is somewhat numerically 
complicated, involving both simulation and iterative estimation 
of the non-linear PARAFAC model. These complexities are a barrier 
to adoption of the method by potential users, such 
as managers of drinking water and wastewater treatment facilities, 
who may only be familiar with or have the software to implement 
simple linear models. 

\begin{figure}
    \centering
    \includegraphics[scale=0.8]{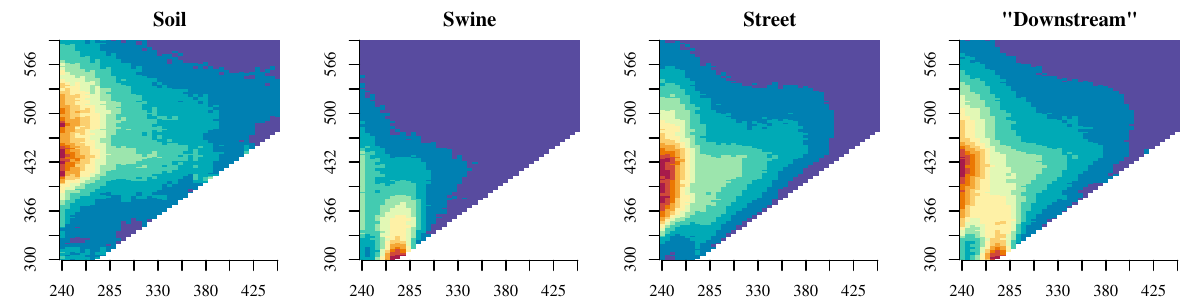}
    \caption{From left to right, three EEMs from the Neuse River dictionary and one hypothetical downstream EEM. As is typical, the lower-right region of each EEM is excluded due to Rayleigh scatter \citep{andersen_practical_2003}. Each of the three left EEMs represents a DOM profile from a particular land use source. The right-most EEM, which is meant to represent an EEM from a downstream water sample, is computed by taking the average of the fluorescence intensities of the left three EEMs.}
    \label{fig:ex_eems}
\end{figure}

As alternatives, \citet{bryan_routine_2023} considered 
two simple methods of source estimation that can be implemented 
using only the tools of 
simple linear regression and vector summation. 
The first method, which we refer to as 
``average-then-regress'' (ATR),  proceeds by averaging the dictionary 
DOM profiles by source category, and then 
regressing the downstream profile on these average profiles. 
The second method, called ``regress-then-sum'' (RTS), 
instead regresses the downstream profile 
on all of the dictionary profiles, then sums the coefficients by 
category. In multiple simulation studies, 
it was observed 
that the RTS method provided notably superior estimates than 
either the FluorMod method or the ATR method. Some heuristic explanation of 
this phenomenon was 
given, but no theory was provided. 

In this article  we formalize the ATR and RTS methods in the 
context of a latent variable model for the downstream DOM profile, which we refer to as the 
\emph{source apportionment model}. Marginalizing over the 
latent variables, this model can be expressed as a linear regression 
model where both the mean and covariance of the downstream DOM profile are affected by the proportion of DOM arising from each source. We show that the ATR estimate corresponds to a feasible 
ordinary least-squares (OLS) estimate, whereas the RTS estimate 
corresponds to a type of feasible generalized least squares (GLS) estimate. 
This result explains the observed superior performance of the 
RTS estimate, as 
GLS estimates have lower mean squared error than OLS 
estimates in general. Additionally, we show how this GLS framework may be used to obtain standard errors for the coefficient estimates, as well as feasible predictors of missing data in the downstream DOM profile. 

While our discussion focuses on fluorescence spectroscopy data, we note that the source apportionment model may be applied to other kinds of data with similar structure, such as hyperspectral images, audio spectrograms, and power meter readings collected over time. For each of these, the estimated coefficients may represent, respectively, proportions of land cover (tree, water, street) in a given pixel \citep{bioucas-dias_hyperspectral_2012}, proportions of note amplitudes sounding at a given time \citep{benetos_automatic_2019}, and proportions of several appliances used over the period of a given month \citep{wytock_contextually_2013}. However, the source apportionment model is distinct from other models that are commonly applied to these data for the purpose of \textit{source separation}. In the task of source separation, the estimands of interest are the unobserved source signals themselves, not the coefficients representing the contributions of these signals to the total.

The remainder of this article is as follows: 
In the next section, we formulate the source apportionment model and 
inference problem, and describe the ATR and RTS estimates of the 
source proportions. 
In Section 3 we show how the ATR and RTS estimates can be interpreted as 
OLS and GLS estimates, respectively, in a linear regression model. We then extend this analogy to propose ATR and RTS predictors for missing data. Section 4 discusses the relative variability of the ATR and RTS estimates and also develops a method to obtain standard errors for the RTS estimates. Finally, Section 5 illustrates the results in a numerical study using the dictionary of 202 DOM profiles originally described in \citet{osburn_predicting_2016}. Directions for further research are discussed in Section 6. 

\section{The source apportionment model} 
Let $\bl y$ be a $p$-dimensional vector representing the DOM profile of a downstream water sample of unknown composition. For such a profile obtained using fluorescence spectroscopy, it is reasonable to assume that 
$\bl y$ is a weighted sum of 
$K$ \emph{latent source profiles}:
$\bl x_1^*,\ldots, \bl x_K^*$,
\begin{align}
\bl y & = \theta_1 \bl x_1^* + \cdots + \theta_K \bl x_K^*.
\label{eqn:llmod} 
\end{align} 
The latent source profiles represent the DOM profiles of the component water samples from each of the $K$ source categories, which contribute to the combined water sample with profile $\bl y$. The vector $\bs\theta$ represents the proportions of each of the $K$ sources that contribute to $\bl y$. If $\bl x_1^*, \dots, \bl x_K^*$ were known, $\bs \theta$ could be determined exactly as the solution to a least squares regression. However, the latent source profiles cannot be observed directly, as only the total downstream DOM profile $\bl y$ can be measured by the spectrometer.

As a substitute for direct observation, we assume each latent source profile is a random vector arising from a source-specific distribution, so that $\bl x_1^* \sim P_1 ,\ldots, \bl x_K^* \sim P_K$ with $\bl x_1^*,\ldots, \bl x_K^*$ being jointly independent. We further assume that data information about $P_1,\ldots, P_K$ is available in the form of a dictionary of $n$ DOM profiles $\bl X \in \mathbb{R}^{p \times n}$. The dictionary profiles may be mixtures of known source proportions in general (see comment at the conclusion of Section \ref{sec:var}), but for now, we assume each dictionary DOM profile is representative of exactly one source category, so that $n = \sum_{k=1}^K n_k$. Letting $\bl x_{i,k}$ be the DOM profile of the $i$th dictionary water sample from source category $k$, the 
$\bl x_{i,k}$'s along with the latent $\bl x_{k}^*$'s are modeled as random samples from the source-specific distributions
\begin{align}
\bl x_{1,k} ,\ldots, \bl x_{n_k,k},\bl x_k^* \sim 
 \text{i.i.d.} \  P_k, \ k=1,\ldots,K,  
\label{eqn:smod} 
\end{align} 
with these profiles additionally being independent across source categories. We refer to the linear model (\ref{eqn:llmod}) together 
with the sampling model (\ref{eqn:smod}) as the 
\emph{source apportionment model}, and refer to the 
task of estimating $\bs\theta$ from $\bl y$ and the dictionary profiles 
as the \emph{source apportionment problem}. In what follows, we consider the source apportionment problem in source apportionment models with $n < p$ and $P_1,\ldots, P_K$ non-degenerate, so that $\bl X$ is full-rank with probability 1.

The source apportionment model bears some resemblance to a latent 
factor model. It may also be viewed a linear regression model with correlated 
errors. 
To make these connections, we first 
write the model in matrix 
form: Let $\bl X^* \in \mathbb R^{p\times K}$ be the matrix 
formed by column-binding $\bl x_1^*,\ldots, \bl x_K^*$, and let 
$\bs\theta = (\theta_1,\ldots, \theta_K)^\top$. Then the sampling model 
(\ref{eqn:llmod}) is $\bl y = \bl X^* \bs\theta$. This looks similar to a linear latent factor model, which expresses 
a data vector as a factor loading matrix multiplied by a latent 
factor vector, both of which are unobserved. However, the source 
apportionment model differs from the latent factor model in terms of 
both the target of inference, and the data available for estimation. 
In particular, in the source apportionment problem there is only one outcome vector $\bl y$, the matrix $\bl X^*$ is viewed as random, and the target of 
inference is $\bs\theta$. In contrast, in factor analysis
we would have multiple $\bl y$-vectors observed, $\bs\theta$ would be viewed as random (typically $\bs\theta \sim N_K(\bl 0, \bl I_K)$), and the target of inference would be $\bl X^*$.  

Now let $\bs \mu_k = \Exp{\bl x_k^*}, \Sigma_k = \Var{\bl x_k^*}$, 
$k=1,\ldots, K$, be the mean vectors and covariance matrices of the 
distributions $P_1,\ldots, P_K$, and let $\bl M\in \mathbb{R}^{p\times K}$ 
be the matrix obtained by column-binding $\bs \mu_1,\ldots, \bs\mu_K$. 
Marginalizing over the latent profile vectors $\bl X^*$, we have 
\begin{align*} 
\Exp{\bl y} & = 
   \Exp{ \bl X^* \bs \theta }  = \bl M \bs\theta \\ % + \Exp{\bs\epsilon } \\ 
\Var{\bl y} & =  \Var{\theta_1 \bl x_1^* + \cdots +  \theta_K \bl x_K^* } \\ 
   & =  \theta_1^2 \Sigma_1 + \cdots +  \theta_K^2   \Sigma_K \equiv
  \Sigma_\theta. 
\end{align*}
If $\bl M$ were known, then the above two equations specify a linear 
regression model for $\bl y$, in which case the 
OLS estimate of $\bs\theta$ would be 
$(\bl M^\top \bl M)^{-1} \bl M^\top \bl y$, and the GLS estimate 
would be 
$(\bl M^\top\Sigma_\theta^{-1} \bl M)^{-1} \bl M^\top\Sigma_\theta^{-1} 
\bl y$. Of course, the latter can only be computed if additionally
$\bs\theta$ were known, which if it were, would make estimation 
unnecessary. If instead
\begin{align}\label{eq:equal_covs}
\Sigma_1 = \cdots = \Sigma_K := \Sigma,
\end{align}
then $\Var{\bl y} \propto \Sigma$, and the GLS estimate can be computed without knowledge of $\bs\theta$ because it is invariant to re-scaling of the error covariance matrix. Note that here and in what follows $\bl A \propto \bl B$ means $\bl A = c \bl B$ for some constant $c$. In the next section, we use assumption \eqref{eq:equal_covs} to develop the ATR and RTS estimates and discuss their respective connections to OLS and GLS estimates.

\section{Linear estimators of source proportions}\label{sec:linear}

The authors in \cite{bryan_routine_2023} proposed two estimates, the ATR and RTS estimates, as solutions to the source apportionment problem. Both estimates can be motivated by the idea that elements of the DOM profile dictionary $\bl X \in \mathbb{R}^{p \times n}$ or functions thereof may serve as surrogates for the latent profiles $\bl X^*$. Let $\bl A$ be the $n \times K$ matrix with entries
\begin{equation}\label{eq:design}
    A_{i k} = \left\{\begin{array}{cl}
         1 & ~~\text{if DOM profile $i$ is from source category $k$} \\
         0 & ~~\text{otherwise.} 
    \end{array}\right. 
\end{equation}
Then the ATR and RTS estimates may be written as
\begin{align*}
    \hat{\bs \theta}_{\mathrm{ATR}} &= \bl A^\top \bl A(\bl A^\top \bl X^\top \bl X \bl A)^{-1} \bl A^\top \bl X^\top \bl y, \\
    \hat{\bs \theta}_{\mathrm{RTS}} &= \bl A^\top (\bl X^\top \bl X)^{-1} \bl X^\top \bl y .
\end{align*}
The ATR estimate is obtained by regressing $\bl y$ on the matrix of DOM profile averages from each source category, which can be written as $\bl X \bl A (\bl A^\top \bl A)^{-1}$. The RTS estimate is obtained by first regressing $\bl y$ on the matrix containing all dictionary profiles and then summing the resulting coefficients by source category. Intuition suggests that the ATR estimate should perform well when each of the latent DOM profiles resembles the average dictionary profile from the corresponding source category. The RTS estimate, on the other hand, may perform well even if each latent profile only resembles one of the dictionary profiles from its source category.

In the context of the source apportionment model, the ATR and RTS estimates can be more formally understood as OLS and GLS estimates of $\bs \theta$, where $\bl X$ has been used to obtain feasible substitutes for the unknowns $\bl M$ and $\Sigma$. According to the sampling model \eqref{eqn:smod}, we have that $\Exp{\bl X} = \bl M \bl A^\top$. Furthermore, by independence of the DOM profiles within and across source categories, along with assumption \eqref{eq:equal_covs}, we have $\Var{\bl X} \propto \bl I_n \otimes \Sigma$, where $\otimes$ denotes the Kronecker product. Let $\hat{\bl M} = \bl X \bl A (\bl A^\top \bl A)^{-1}$, so that $\hat{\bl M}$ is the OLS estimate of $\bl M$ based on $\bl X$. Next, let
\begin{align*}
\bl S = (\bl X - \hat{\bl M} \bl A^\top) (\bl X - \hat{\bl M} \bl A^\top)^\top = \bl X (\bl I_n - \bl P_{\bl A}) \bl X^\top,
\end{align*}
where $\bl P_{\bl A} = \bl A (\bl A^\top \bl A)^{-1} \bl A^\top$, so that $\bl S$ is the $p \times p$ residual sum of squares matrix from the OLS fit. Finally, define a mean-zero ``residual" matrix $\bl E = \bl X \bl N$, where $\bl N \in \mathbb{R}^{n \times (n - K)}$ is an orthonormal basis for the null space of $\bl A^\top$. Then $\bl N^\top \bl N = \bl I_{n-K}$, $\bl N \bl N^\top = \bl I_n - \bl P_{\bl A}$, $\bl S = \bl E \bl E^\top$, and
\begin{align*}
\Exp{\hat{\bl M}} = \Exp{\bl X} \bl A (\bl A^\top \bl A)^{-1} = \bl M \\
\Exp{\bl S} = \Exp{\bl E \bl E^\top} = \Sigma \times (n - K) .
\end{align*}
Hence, a reasonable, feasible OLS estimate of $\bs \theta$ is $(\hat{\bl M}^\top \hat{\bl M})^{-1} \hat{\bl M}^\top \bl y$, which is precisely the ATR estimator, as the columns of $\hat{\bl M}$ are the average DOM profiles from each source category. Since $n <p$, $\bl S$ will be singular, so we consider estimating $\Sigma$ as
\begin{align*}
    \hat{\Sigma}_\gamma \propto \bl S + \gamma \bl I_p 
\end{align*}
for some regularization parameter $\gamma \geq 0$. This leads to the feasible GLS estimate
\begin{align}\label{eq:fgls}
    \hat{\bs \theta}_{\gamma} = (\hat{\bl M}^\top \hat{\Sigma}_\gamma^{-1} \hat{\bl M})^{-1} \hat{\bl M}^\top \hat{\Sigma}_\gamma^{-1} \bl y .
\end{align}

To see the relationship between this estimate and the RTS estimate, consider that the RTS estimate can be written as $\hat{\bs \theta}_{\mathrm{RTS}} = \bl A^\top (\bl X^\top \bl X)^{-1} \bl X^\top \bl y$. This is just $\bl A^\top \hat{\bs \beta}$, where $\hat{\bs \beta}$ is the OLS estimate for $\bs \beta$ in the linear model $\Exp{\bl y} = \bl X \bs \beta$. We can reparameterize this linear model as
\begin{align*}\label{eq:reparam}
    \bl X \bs \beta &= \bl X \bl P_{\bl A} \bs \beta + \bl X (\bl I_n - \bl P_{\bl A}) \bs \beta \\
    &= \hat{\bl M} \bs \theta + \bl E \bs \eta \\
    &\equiv \bl Z \bs \psi,
\end{align*}
where $\bs \theta = \bl A^\top \bs \beta$, $\bs \eta = \bl N^\top \bs \beta$, $\bl Z = [\hat{\bl M} ~ \bl E]$, and $\bs \psi = (\bs \theta^\top \bs \eta^\top)^\top$. This reparameterization relates the linear model with all DOM profiles as regressors to the linear model with two sets of regressors: the average profiles $\hat{\bl M}$ and the ``residual" profiles $\bl E$. The following proposition then relates the coefficient estimates from these two models.
\begin{proposition}
    \label{prop:1}
    Let $\hat{\bs \beta} = (\bl X^\top \bl X)^{-1} \bl X^\top \bl y$, and let $\hat{\bs \psi} = (\bl Z^\top \bl Z)^{-1} \bl Z^\top \bl y$, where $\bl Z$ is defined as above. Then
    \begin{equation*}
        \bl A^\top \hat{\bs \beta} = \hat{\bs \psi}[1:K]
    \end{equation*}
\end{proposition}

Hence, $\hat{\bs \theta}_{\mathrm{RTS}} = \hat{\bs \psi}[1:K]$. Together with previous results, Proposition \ref{prop:1} implies that while the ATR estimate is equivalent to the OLS estimate of the $\hat{\bl M}$ coefficients in a regression of $\bl y$ on $\hat{\bl M}$, the RTS estimate is equivalent to the OLS estimate of the $\hat{\bl M}$ coefficients in an expanded linear model, one which controls for the variation described by $\bl E$. On the other hand, applying \cite{seber_linear_2003} Theorem 3.6(i) to $\hat{\bs \psi}[1:K]$, we find another equivalence 
\begin{align*}
    \hat{\bs \theta}_{\mathrm{RTS}} = (\hat{\bl M}^\top (\bl I_p - \bl P_{\bl E}) \hat{\bl M})^{-1} \hat{\bl M}^\top (\bl I_p - \bl P_{\bl E}) \bl y,
\end{align*}
where $\bl P_{\bl E} = \bl E (\bl E^\top \bl E)^{-1} \bl E^\top$. This shows that the RTS estimate can also be expressed as a type of feasible GLS estimate, in which $\bl I_p - \bl P_{\bl E}$ plays the role of an inverse covariance matrix. The next proposition makes this notion precise, showing that the RTS estimate is a limiting case of the feasible GLS estimate in \eqref{eq:fgls}.
\begin{proposition}\label{prop:2}
Let $\hat{\bs \theta}_{\gamma} = (\hat{\bl M}^\top \hat{\Sigma}_\gamma^{-1} \hat{\bl M})^{-1} \hat{\bl M}^\top \hat{\Sigma}_\gamma^{-1} \bl y$. Then
    \begin{equation*}
        \lim_{\gamma \rightarrow 0}~ \hat{\bs \theta}_{\gamma} = \hat{\bs \theta}_{\mathrm{RTS}}.
    \end{equation*}
    Also
    \begin{equation*}
        \lim_{\gamma \rightarrow \infty}~ \hat{\bs \theta}_{\gamma} = \hat{\bs \theta}_{\mathrm{ATR}}.
    \end{equation*}
\end{proposition}

The least squares connections developed above can also be extended to the problem of prediction in the source apportionment model. Suppose that instead of observing all elements of the $p$-dimensional DOM profile $\bl y$, we only observe a $q$-dimensional subvector. In a practical setting, this might happen if a downstream DOM profile is measured only on a subset of the excitation frequencies used to measure the profiles in the dictionary. Partition $\bl y$ into observed and unobserved components $\bl y_0 \in \mathbb{R}^{p-q}, \bl y' \in \mathbb{R}^q$, and consider the problem of predicting $\bl y'$ from $\bl X$ and $\bl y_0$.

Assume that $\bl y$ follows the partitioned source apportionment model
\begin{equation*}
\begin{aligned}
    (\Exp{\bl y_0}, \Exp{\bl y'}) &= (\bl M_0 \bs \theta, \bl M' \bs \theta) \\
    \Var{\bl y}  &= \left[ \begin{array}{cc}
         \Sigma_0 & \Delta \\
         \Delta^\top & \Sigma'
    \end{array}\right],
\end{aligned}
\end{equation*}
where $\Sigma_0$ is $(p-q)\times (p-q)$, $\Delta$ is $q\times (p-q)$, and $\Sigma'$ is $q \times q$. The best linear unbiased predictor for the unobserved portion of the downstream profile is then 
\begin{equation}\label{eq:blup}
    \hat{\bl y}' = \bl M' \hat{\bs \theta} + \Delta^\top \Sigma_0^{-1}(\bl y_0 - \bl M_0 \hat{\bs \theta}),
\end{equation}
where $\hat{\bs \theta} = (\bl M_0^\top \Sigma_0^{-1} \bl M_0)^{-1} \bl M_0^\top \Sigma_0^{-1} \bl y_0$ \citep{kariya_generalized_2004}. As before, we can obtain a feasible version of $\hat{\bl y}'$ by using the dictionary of DOM profiles to create substitutes for the unknowns in \eqref{eq:blup}. Partition the dictionary in the same manner as $\bl y$, and let
\begin{equation*}
\begin{aligned}
\hat{\bl M} &= (\hat{\bl M}_0, \hat{\bl M}') \\
\hat{\Sigma}_{\gamma} &= \left[ \begin{array}{cc}
         \hat{\Sigma}_{0 \gamma} & \hat{\Delta} \\
         \hat{\Delta}^\top & \hat{\Sigma}_\gamma'
    \end{array}\right]
\end{aligned}
\end{equation*}
be the corresponding partitions of $\hat{\bl M}, \hat{\Sigma}_{\gamma}$. Then define the feasible predictor
\begin{equation*}
    \hat{\bl y}_\gamma' = \hat{\bl M}' \hat{\bs \theta}_{0 \gamma} + \hat{\Delta}^\top \hat{\Sigma}_{0\gamma}^{-1}(\bl y_0 - \hat{\bl M}_0 \hat{\bs \theta}_{0 \gamma}),
\end{equation*}
where $\hat{\bs \theta}_{0 \gamma} = (\hat{\bl M}_0^\top \hat{\Sigma}_{0\gamma}^{-1} \hat{\bl M}_0)^{-1} \hat{\bl M}_0^\top \hat{\Sigma}_{0\gamma}^{-1} \bl y_0$. In analogy to Proposition \ref{prop:2}, we obtain simple limiting expressions for $\hat{\bl y}_\gamma'$.

\begin{proposition}
\begin{equation*}
    \lim_{\gamma \rightarrow 0}~ \hat{\bl y}_\gamma' = \bl X' (\bl X_0^\top \bl X_0)^{-1} \bl X_0^\top \bl y_0.
\end{equation*}
Also,
\begin{equation*}
    \lim_{\gamma \rightarrow \infty}~ \hat{\bl y}_\gamma' = \hat{\bl M}' (\hat{\bl M}_0^\top \hat{\bl M}_0)^{-1} \hat{\bl M}_0^\top \bl y_0 .
\end{equation*}
\end{proposition}
We call these limiting predictors the RTS and ATR predictors, respectively. Like the ATR and RTS estimates, the ATR and RTS predictors can be explained in simple terms and can be computed using standard linear regression tools.

\section{Variability of the RTS estimates}\label{sec:var}

Because the RTS estimate is a feasible version of the actual GLS estimate in the source apportionment model, it is not, in general, optimal in terms of mean squared error, nor is it necessarily unbiased. However, the same can be said of the ATR estimate. The properties of each depend on the extent to which the feasible approximations $\hat{\bl M} \approx \bl M$ and $\hat{\Sigma}_{\gamma} \approx \Sigma$ hold. In this section, we first analyze the variability of the ATR and RTS estimates assuming an idealized model where the mean and variance of the downstream DOM profile can be described exactly using the dictionary profiles. Doing so offers some insight into when the RTS estimate may be less variable than the ATR estimate and allows us to derive simple standard errors for the RTS estimate. We then discuss what happens to the RTS standard errors in the general case. The variability we consider here is with respect to the variability in the latent DOM profiles only, as we assume the dictionary profiles to be fixed at their observed values.

Suppose that $\bl y$ follows a source apportionment model given by
\begin{equation}\label{eq:idealmod}
\begin{aligned}
    \Exp{\bl y} &= \hat{\bl M} \bs \theta \\
    \Var{\bl y} &= \| \bs \theta \|_2^2 \hat{\Sigma}_{\gamma},
\end{aligned}
\end{equation}
where $\hat{\bl M}, \hat{\Sigma}_{\gamma}$ are functions, as defined in the previous section, of a non-random dictionary $\bl X$ and a design-like matrix $\bl A$. In this model, both $\hat{\bs \theta}_{\mathrm{ATR}}$ and $\hat{\bs \theta}_{\mathrm{RTS}}$ are unbiased since they are both of the form $\bl C^\top \bl y$ for some matrix $\bl C \in \mathbb{R}^{p \times K}$ such that $\bl C^\top \hat{\bl M} = \bl I_K$. However, their variances differ. Recalling the definition of $\hat{\Sigma}_{\gamma}$, we have
\begin{equation}\label{eq:vars}
\begin{aligned}
    \Var{\hat{\bs \theta}_{\mathrm{ATR}}} &\propto (\hat{\bl M}^\top \hat{\bl M})^{-1} \hat{\bl M}^\top \bl S \hat{\bl M} (\hat{\bl M}^\top \hat{\bl M})^{-1} + \gamma (\hat{\bl M}^\top \hat{\bl M})^{-1} \\
    \Var{\hat{\bs \theta}_{\mathrm{RTS}}} &\propto  \gamma \bl A^\top (\bl X^\top \bl X)^{-1} \bl A, \\
\end{aligned}
\end{equation}
with respect to the same constant of proportionality. In the expression for the variance of the RTS estimate, the additional term involving $\bl S$ has vanished because $\bl S = \bl E \bl E^\top$, and 
\begin{equation*}
    \bl A^\top (\bl X^\top \bl X)^{-1} \bl X^\top \bl E = \bl A^\top (\bl X^\top \bl X)^{-1} \bl X^\top \bl X \bl N = \bl A^\top \bl N = \bl 0 .
\end{equation*}
Looking at \eqref{eq:vars}, it is clear that, at least in this idealized model, the variance of the RTS estimate can become arbitrarily small as $\gamma \rightarrow 0$. However, as a consequence of the matrix Cauchy-Schwarz inequality \citep{marshall_matrix_1990} (alternatively, a consequence of the Gauss-Markov Theorem \citep{aitken_ivleast_1936}), we have the following correspondence in the Loewner partial order
\begin{equation*}
    (\hat{\bl M}^\top \hat{\bl M})^{-1} \preceq \bl A^\top (\bl X^\top \bl X)^{-1} \bl A,
\end{equation*}
so the variance of the RTS estimate actually becomes greater than that of the ATR estimate as $\gamma \rightarrow \infty$. The next proposition gives the interval of values for $\gamma$ in which the RTS estimate outperforms the ATR estimate as a function of the various matrices in \eqref{eq:vars}.
\begin{proposition}\label{prop:atrvsrts}
    Let $\gamma >0$, and assume $\bl y$ follows the source apportionment model in \eqref{eq:idealmod}. Let
    \begin{equation*}
    \begin{aligned}
        \bl V_1 &= \bl A^\top (\bl X^\top \bl X)^{-1} \bl A - (\hat{\bl M}^\top \hat{\bl M})^{-1} \\
        \bl V_2 &= (\hat{\bl M}^\top \hat{\bl M})^{-1} \hat{\bl M}^\top \bl S \hat{\bl M} (\hat{\bl M}^\top \hat{\bl M})^{-1}.
    \end{aligned}
    \end{equation*}
    Then $\Var{\hat{\bs \theta}_{\mathrm{RTS}}} \preceq \Var{\hat{\bs \theta}_{\mathrm{ATR}}}$ if and only if $\gamma \leq \lambda_{\min}(\bl V_1^{-1} \bl V_2)$, where $\lambda_{\min}$ denotes the minimum eigenvalue.
\end{proposition}
While the condition in Proposition \ref{prop:atrvsrts} cannot be checked directly because $\gamma$ is unknown, one can determine the range of $\gamma$ values favorable to the RTS estimate because $\lambda_{\min}(\bl V_1^{-1} \bl V_2)$ may be computed from the dictionary. The matrix $\bl V_1$ quantifies the gap between the RTS and ATR variance in the case of entirely isotropic error, and as the scale of this term grows, the region favorable to the RTS estimate shrinks. Recalling that $\bl S = \bl X (\bl I_n - \bl P_{\bl A}) \bl X^\top$, it can be shown that computing the matrix $\bl V_2$ is equivalent to first computing ATR coefficients on each dictionary element and then taking the sum of the source-wise covariance matrices of these coefficients. As the scale of $\bl V_2$ increases, the region favorable to the RTS estimate grows.

The expression for $\Var{\hat{\bs \theta}_{\mathrm{RTS}}}$ in \eqref{eq:vars} suggests that the matrix of squared standard errors
\begin{equation}\label{eq:sse}
    \mathrm{SSE}[\hat{\bs \theta}_{\mathrm{RTS}}] = \frac{(\bl y^\top (\bl I_p - \bl P_{\bl X}) \bl y)}{p - n} \bl A^\top (\bl X^\top \bl X)^{-1} \bl A
\end{equation}
may be used to estimate the variability of the RTS estimate. Assuming \eqref{eq:idealmod},  $\mathrm{SSE}[\hat{\bs \theta}_{\mathrm{RTS}}]$ is unbiased for $\Var{\hat{\bs \theta}_{\mathrm{RTS}}}$ because $\frac{(\bl y^\top (\bl I - \bl P_{\bl X}) \bl y)}{p - n}$ is an unbiased estimate of the magnitude of the isotropic component of the variance of $\bl y$. Specifically,
\begin{align*}
    \Exp{(\bl I_p - \bl P_{\bl X}) \bl y} &= (\bl I_p - \bl P_{\bl X}) \hat{\bl M} \bs \theta = \bl 0 \\
    \Var{(\bl I_p - \bl P_{\bl X}) \bl y} &= \| \bs \theta \|_2^2 (\bl I_p - \bl P_{\bl X}) \hat{\Sigma}_{\gamma} = c \gamma \| \bs \theta \|_2^2 (\bl I_p - \bl P_{\bl X}),
\end{align*}
for some constant $c > 0$. Therefore, $\Exp{(\bl y^\top (\bl I_p - \bl P_{\bl X}) \bl y)/(p-n)} = c \gamma \| \bs \theta \|_2^2$. When \eqref{eq:idealmod} does not hold, the difference between $\mathrm{SSE}[\hat{\bs \theta}_{\mathrm{RTS}}]$ and $\Var{\hat{\bs \theta}_{\mathrm{RTS}}}$ will depend on the relationship between the dictionary profiles and the unknowns $\bl M, \Sigma$. The next proposition characterizes the average behavior of $\mathrm{SSE}[\hat{\bs \theta}_{\mathrm{RTS}}]$ in a general source apportionment model in terms of three mutually orthogonal subspaces of $\mathbb{R}^p$.

\begin{proposition}\label{prop:sebias}
    Assume that $\bl y$ follows a general source apportionment model
    \begin{equation*}
        \begin{aligned}
        \Exp{\bl y} &= \bl M \bs \theta \\
        \Var{\bl y} &= \|\bs \theta \|_2^2 \Sigma .
        \end{aligned}
    \end{equation*}
    Let $v_k$ be the $k^{\text{th}}$ diagonal entry of $\Var{\hat{\bs \theta}_{\mathrm{RTS}}}$ and let $\hat{v}_k$ be the $k^{\text{th}}$ diagonal entry of $\mathrm{SSE}[\hat{\bs \theta}_{\mathrm{RTS}}]$. Also let
    \begin{itemize}
        \item $\bl U_1$ be the $p \times (n-k)$ matrix whose columns are the left singular vectors of $\bl E$.
        \item $\bl U_2$ be the $p \times k$ matrix whose columns are the left singular vectors of $\bl X (\bl X^\top \bl X)^{-1} \bl A$.
        \item $\bl U_3$ be the $p \times (p - n)$ matrix whose columns are the left singular vectors of $\bl P_{\bl X}$.
    \end{itemize}
    Then
    \begin{align}
     \Exp{v_k - \hat{v}_k} / \|\bs \theta \|_2^2 &\leq \bl a_k^\top (\bl X^\top \bl X)^{-1} \bl a_k[\lambda_{\max}(\bl U_2^\top \Sigma \bl U_2) - \bar{\lambda}(\bl U_3^\top \Sigma \bl U_3)]\label{line:ub} \\
     \Exp{v_k - \hat{v}_k} / \|\bs \theta \|_2^2 &\geq \bl a_k^\top (\bl X^\top \bl X)^{-1} \bl a_k[\lambda_{\min}(\bl U_2^\top \Sigma \bl U_2) - \bar{\lambda}(\bl U_3^\top \Sigma \bl U_3) - \bar{\lambda}(\bl U_3^\top \bl M \bl M^\top \bl U_3)] \label{line:lb}
    \end{align}
    where $\lambda_{\min}$ denotes the minimum eigenvalue, $\lambda_{\max}$ denotes the maximum eigenvalue, and $\bar{\lambda}$ denotes the average eigenvalue.
\end{proposition}

To interpret Proposition \ref{prop:sebias}, first note that the columns of $\bl U_1$ form an orthonormal basis for the subspace spanned by the first $n-k$ singular vectors of $\hat{\Sigma}_{\gamma}$. The magnitude of $\Sigma$ lying in this subspace, the principal subspace of the feasible approximation to $\Sigma$, contributes nothing to the bias of the squared standard errors. Instead, the bias depends on the magnitude of $\Sigma$ lying along the remaining, orthogonal, directions described by the columns of $\bl U_2$ and $\bl U_3$. The term $\bar{\lambda}(\bl U_3^\top \bl M \bl M^\top \bl U_3)$ is the average squared residual between $\bl M$ and its projection onto the column space of $\bl X$. The more the columns of $\bl M$ lie in the column space of the dictionary DOM profiles, the more this term approaches 0. The upper bound in \eqref{line:ub} can be interpreted as measuring both the non-isotropy of $\bl U_2^\top \Sigma \bl U_2$ and the difference between the magnitudes of $\bl U_2^\top \Sigma \bl U_2$ and $\bl U_3^\top \Sigma \bl U_3$. The lower bound in \eqref{line:lb} (excluding $\bar{\lambda}(\bl U_3^\top \bl M \bl M^\top \bl U_3)$) has the same interpretation. If the variability described by $\Sigma$ in the directions orthogonal to $\bl U_1$ is isotropic and the columns of $\bl M$ can be written as linear combinations of the columns of $\bl X$, then $\mathrm{SSE}[\hat{\bs \theta}_{\mathrm{RTS}}]$ will be unbiased for $\Var{\hat{\bs \theta}_{\mathrm{RTS}}}$. This is of course the case when $\bl M = \hat{\bl M}$ and $\Sigma = \hat{\Sigma}_{\gamma}$.

To conclude, we note that none of the analyses in this section nor those in the previous section depend on the particular structure of $\bl A$ described in \eqref{eq:design}. In the Neuse River dataset, each dictionary EEM comes from a single land use source, so the corresponding $\bl A$ contains only zeros and ones. However, experimental conditions in other source apportionment problems may permit the collection of a dictionary with mixed elements of known source proportions. The rows of the corresponding $\bl A$ will then be composed of these known mixing proportions, and the corresponding ``RTS" coefficients will be weighted sums of regression coefficients.

\section{Source apportionment in practice}\label{sec:data}

While least squares theory guarantees the superiority of the GLS estimate and predictor over those of OLS, the results from Section \ref{sec:var} show that the relative performance of the corresponding RTS and ATR quantities is model- and dictionary-dependent. In this section, we provide numerical evidence that the RTS estimate and predictor are indeed superior to the ATR estimate and predictor in the context of fluorescence spectroscopy measurements of DOM. The approach we take is to evaluate the properties of the ATR and RTS methods in the context of a realistic source apportionment model, for which the population-level quantities $\bl M$ and $\Sigma$ are derived from the 202 DOM profiles in the Neuse River dataset.

Let $\bl{X}$ be the $4891 \times 202$ matrix whose columns are the DOM profiles in the Neuse River dataset. The numerical results in this section are computed with respect to a source apportionment model for a downstream DOM profile that has mean and covariance
\begin{equation*}
\begin{aligned}
    \bl M \bs \theta &= \bl X \bl A (\bl A^{\top} \bl A)^{-1} \bs \theta \\
    \|\bs \theta\|_2^2 \Sigma &= \|\bs \theta \|_2^2 \left[ \frac{\nu^*}{n-K} \bl X (\bl I_n - \bl P_{\bl A}) \bl X^\top + \gamma^* \bl I_p \right],
\end{aligned}
\end{equation*}
where $\nu^*$ and $\gamma^*$ are positive scalars chosen to make $\Sigma$ equal to the optimal covariance estimator defined in \cite{ledoit_well-conditioned_2004} Eqn. 14. Within this model, we evaluate the ATR and RTS methods for each of 250 values of $\bs \theta$, and for each of 4 possible dictionaries, yielding a total of 1000 distinct points of evaluation. The $\bs \theta$ values, pictured in Figure \ref{fig:thetas}, are simulated independently from a $\mathrm{Dirichlet}(\bl 1_K / K)$ distribution, which has coverage near the boundaries of the $K$-dimensional probability simplex. Each dictionary is constructed by selecting a fraction, $\alpha$, of the DOM profiles from the total Neuse River dataset. The profiles are sampled uniformly at random from each source category, and then column-bound to form a dictionary matrix $\bl X_{\alpha}$ for each of $\alpha \in \{0.25, 0.5, 0.75, 0.95\}$.

\begin{figure}
    \centering
    \includegraphics[scale=0.45]{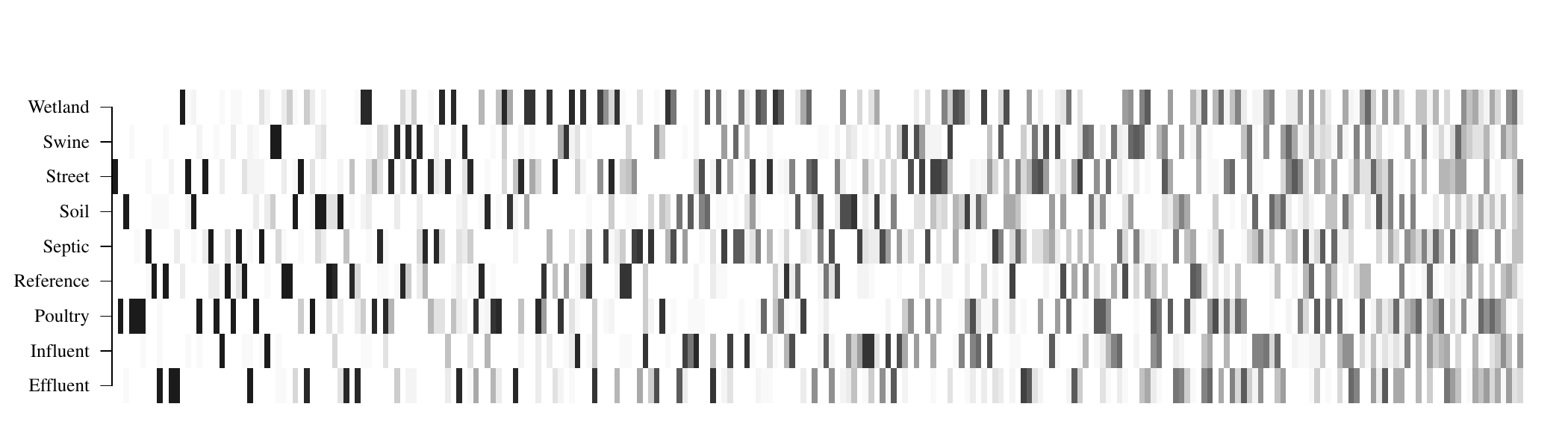}
    \caption{Depiction of $\bs{\theta}$ values at which ATR and RTS estimates and predictors are evaluated. There are $K=9$ source categories, each corresponding to a land use source. Darker/lighter gray signifies higher/lower coefficient weight. The $\boldsymbol{\theta}$'s all have positive entries that sum to $1$ and are arranged in increasing order of entropy, from left to right.}
    \label{fig:thetas}
\end{figure}

As $\alpha$ increases, the dictionary matrix $\bl X_{\alpha}$ explains more of the variation in the population mean and covariance of the source apportionment model considered in this study. The different values of $\alpha$ therefore allow us to observe what happens as the feasible ATR and RTS estimates approach the oracle OLS and GLS estimates, which require knowledge of $\bl M$ and $\Sigma$. To be precise, as $\alpha \rightarrow 1$ the properties of the RTS estimate actually approach those attained in the ideal model \eqref{eq:idealmod}. However, as the optimal $\gamma^*$ in this study is quite small, these are very close to the properties of the oracle GLS estimates.

\begin{figure}
    \centering
    \includegraphics[scale=0.6]{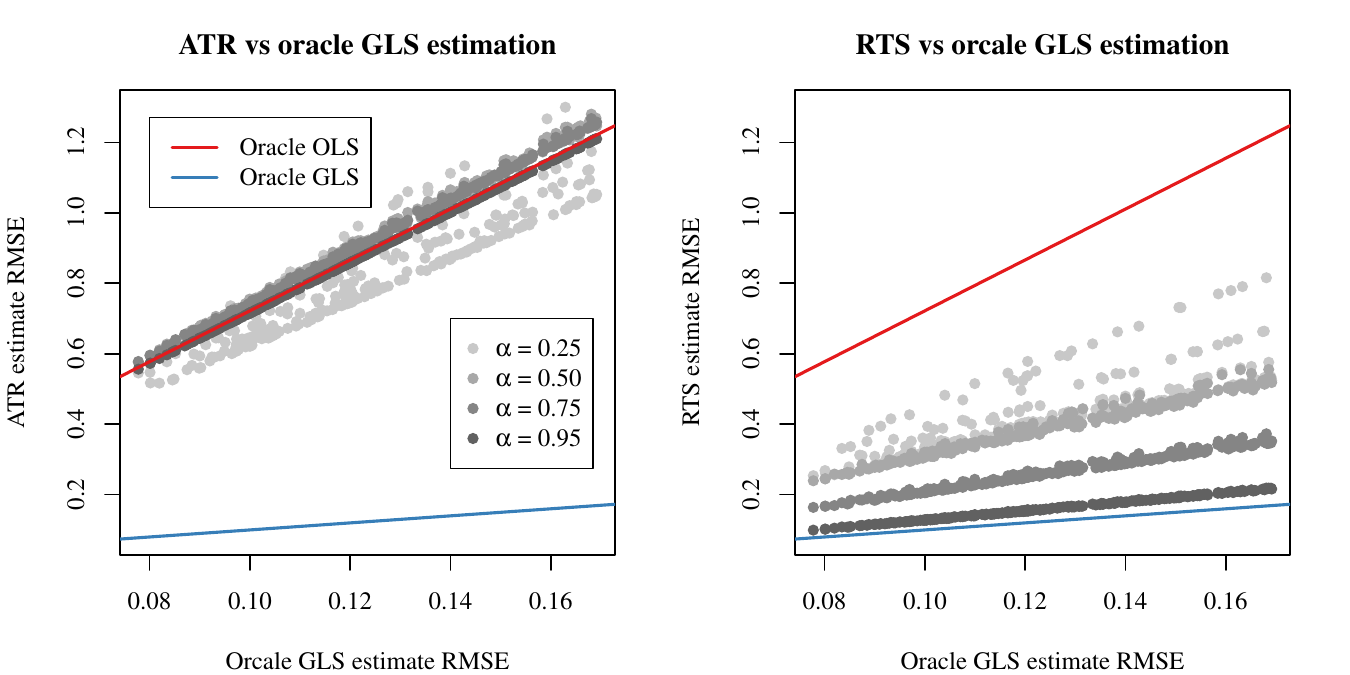}
    \caption{Performance of the ATR and RTS estimates in the numerical study as measured by RMSE. Each point corresponds to a different value of $\bs \theta$ and $\alpha$. The RMSE of the RTS estimates is lower than that of the ATR estimates for all values of $\bs \theta$ and $\alpha$. As $\alpha$ increases, the RTS RMSE continues to decrease to that of oracle GLS, while the ATR RMSE stabilizes around that of oracle OLS.}
    \label{fig:bias_variance}
\end{figure}

As seen in Figure \ref{fig:bias_variance}, the square root of the mean squared error (RMSE), $\sqrt{\mathrm{E}[\|\hat{\bs{\theta}} - \bs{\theta}\|_2^2]}$, of the ATR and RTS estimates gets closer to the RMSE attained by their oracle counterparts as $\alpha$ increases. Importantly, the RMSE of the RTS estimate is lower than that of the ATR estimate for all values of $\bs \theta$, even when $\alpha = 0.25$, suggesting that the RTS estimate should be preferred to the ATR estimate when applied to fluorescence profiles of DOM. The fact that the points in Figure \ref{fig:bias_variance} seem to lie nearly along straight lines suggests that the RMSE of all of these estimates is dominated by the variance component (the variances of the ATR, RTS, oracle OLS, and oracle GLS estimates are all proportional to $\| \bs \theta \|$). This suggests that the superiority of the RTS estimate relative to the ATR estimate is due primarily to a reduction in variance, which is consistent with its connections to GLS.

\begin{figure}
    \centering
    \includegraphics[scale=0.65]{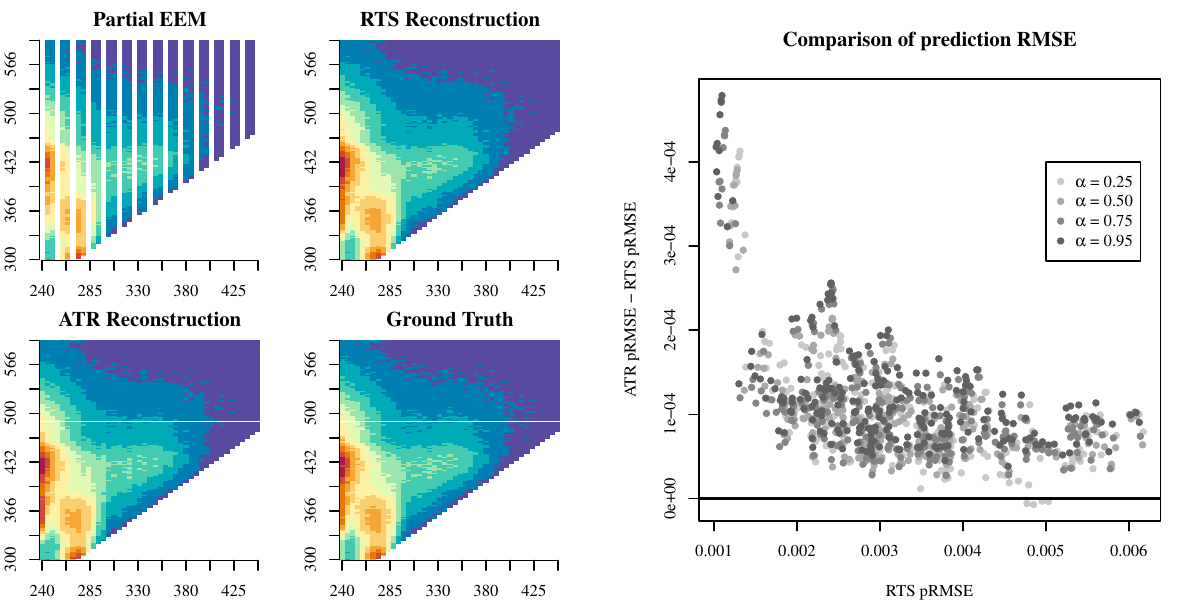}
    \caption{Performance of the ATR and RTS predictors. On the left, raster images of an example partial, ATR/RTS reconstructed, and true simulated EEM. On the right, the $y$-axis is the difference between the prediction RMSE of the ATR predictor and that of the RTS predictor. The $x$-axis is the prediction RMSE of the RTS predictor. Each point corresponds to a different value of $\bs \theta$ and $\alpha$.}
    \label{fig:prediction}
\end{figure}

A similar phenomenon is seen regarding the ATR and RTS predictors. To assess these, we imagine that a downstream DOM profile is only scanned at 28 out of the 43 excitation wavelengths used to create the Neuse River dataset (see Figure \ref{fig:prediction}, left). As described in Section 3, the prediction task is then to reconstruct the full profile from the partially observed downstream profile and the fully observed dictionary profiles. Using the same evaluation points as before, we compute prediction RMSE $\sqrt{\Exp{\|\bl y' - \hat{\bl{y}}'\|_2^2}}$ for each value of $\bs \theta$ and $\alpha$. While visually there are only minor differences between the ATR and RTS reconstructions in the left panel of Figure \ref{fig:prediction}, it is clear from the right panel of Figure \ref{fig:prediction} that the RTS predictor has lower prediction RMSE than the ATR predictor for nearly all values of $\bs \theta$ and $\alpha$.

\begin{figure}
    \centering
    \includegraphics[scale=0.6]{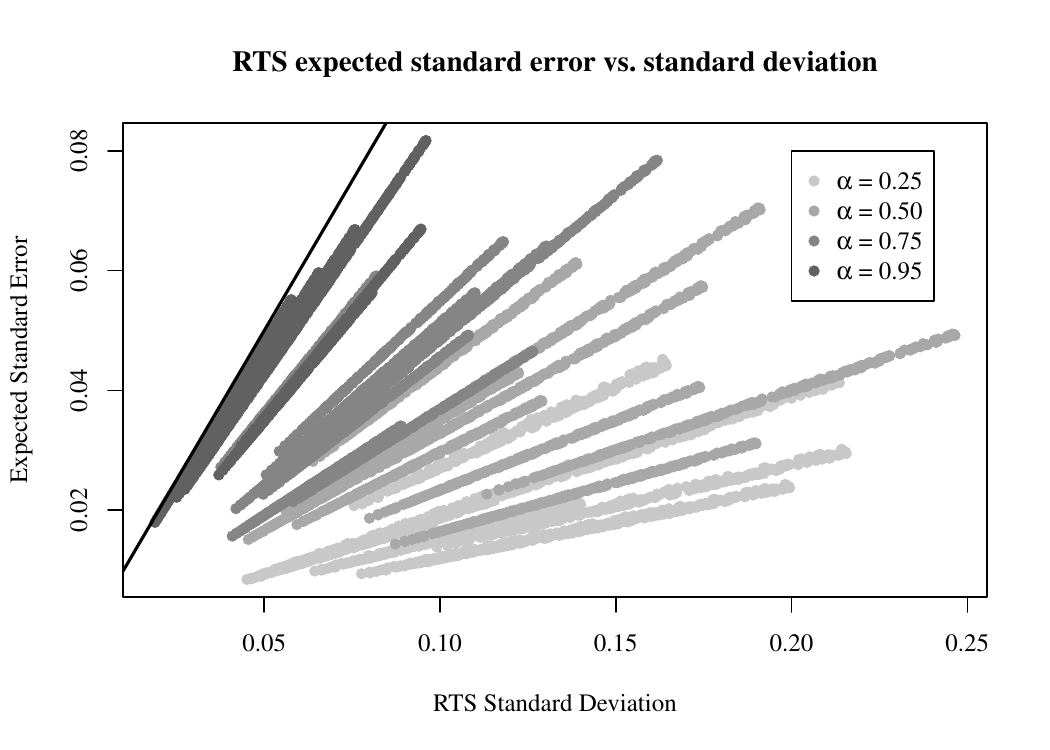}
    \caption{Expected standard error of the RTS estimate versus its standard deviation. There is one point per value of $\bs \theta$, $\alpha$, and source category, which produces the effect of having $K=9$ visually distinct trajectories for each value of $\alpha$. Our standard errors are biased downwards for all values of $\alpha$, though the bias decreases as $\alpha$ increases.}
    \label{fig:se}
\end{figure}

The results concerning our proposed standard errors for the RTS estimate are pictured in Figure \ref{fig:se}. The standard errors are biased downwards, meaning they tend to underestimate the standard deviation of the RTS estimate, for all values of $\alpha$. At $\alpha = 0.25$, the scale of the bias is quite large relative to the standard deviation. However, at $\alpha = 0.75$ and $\alpha = 0.95$ the standard errors begin to give a more accurate sense of the true variability of the RTS estimate. When $\alpha = 1$ (not pictured), the standard errors are unbiased, as discussed in Section \ref{sec:var}.

\section{Discussion} \label{sec:discussion}

The source apportionment model is a latent variable model for DOM profiles collected downstream of known land-use sources. In the context of the source apportionment model, least squares theory implies the existence of an optimal linear estimate of source proportions, the GLS estimates, which requires knowledge of the source-specific mean DOM profiles and covariance matrices. Given a dictionary of DOM profiles collected from the same land-use sources that contribute to the downstream profile, a feasible version of this optimal estimate, the RTS estimate, may be computed using the tools of simple linear regression. While the RTS estimate is not guaranteed to be optimal in the source apportionment model, our numerical results suggest that the RTS estimate has similar behavior to its oracle GLS counterpart when applied to fluorescence spectroscopy measurements of DOM. Similarly, the RTS predictor behaves like the oracle GLS predictor.

As discussed in Section \ref{sec:var}, the bias in our proposed RTS squared standard errors results from a discrepancy between the matrices $\bl M$ and $\Sigma$ and their dictionary-derived approximations. A promising direction for debiasing these squared standard errors is to try to estimate the bias components in Proposition \ref{prop:sebias} from the dictionary, perhaps using disjoint subsets of the dictionary profiles. However, a full account of such an approach should consider the randomness in the dictionary profiles, and remains a direction for future research. 

The classical least squares framework used in this article provides a prescription for how to use a DOM profile dictionary to solve the source apportionment problem. However, it ignores the non-negative nature of fluorescence spectroscopy data, and places no non-negativity restrictions on the estimated source proportions. Another interesting research direction is to study the properties of positive analogues of the OLS, GLS, ATR, and RTS estimates computed using non-negative least squares regression \citep{lawson_solving_1995}, and to determine the extent to which the results of classical regression still apply.

The Neuse River dataset, proofs of the propositions in this article, and software to replicate the figures in this article are available as supplementary files.

\bibliographystyle{chicago}

\bibliography{refs}

%%%%%%%%%%%%%%%%%%%%%%%%%%%%%%%%%%%%%%%%%%%%%%
%% Example with single Appendix:            %%
%%%%%%%%%%%%%%%%%%%%%%%%%%%%%%%%%%%%%%%%%%%%%%

\pagebreak

\section{Proofs}

\begin{proof}[Proof of Proposition \ref{prop:1}]
Recall that $\bl N \bl N^\top = \bl I_n - \bl P_{\bl A}$, $\bl A^\top \bl N = \bl 0$. From these, we derive
\begin{align*}
    \bl A^\top \left[ \begin{array}{cc}
         \bl A (\bl A^\top \bl A)^{-1} & \bl N
    \end{array}\right] &= \left[ \begin{array}{cc}
         \bl I_K & \bl 0
    \end{array}\right]
\end{align*}
and
\begin{align*}
    \left[ \begin{array}{cc}
         \bl A (\bl A^\top \bl A)^{-1} & \bl N
    \end{array}\right]\left[ \begin{array}{cc}
         \bl A & \bl N
    \end{array}\right]^\top = \bl I_n .
\end{align*}
Therefore, 
\begin{align*}
    \hat{\bs \psi}[1:K] &= \left[ \begin{array}{cc}
         \bl I_K & \bl 0
    \end{array}\right] \hat{\bs \psi} \\
    &= \bl A^\top \left[ \begin{array}{cc}
         \bl A (\bl A^\top \bl A)^{-1} & \bl N
    \end{array}\right] (\bl Z^\top \bl Z)^{-1} \bl Z^\top \bl y \\
    &= \bl A^\top \left[ \begin{array}{cc}
         \bl A (\bl A^\top \bl A)^{-1} & \bl N
    \end{array}\right] \left(\left[ \begin{array}{c}
         \bl (\bl A^\top \bl A)^{-1} \bl A^\top \\
         \bl N^\top
    \end{array}\right] \bl X^\top \bl X \left[ \begin{array}{cc}
         \bl A (\bl A^\top \bl A)^{-1} & \bl N
    \end{array}\right]\right)^{-1} \bl Z^\top \bl y \\
    &= \bl A^\top \left[ \begin{array}{cc}
         \bl A (\bl A^\top \bl A)^{-1} & \bl N
    \end{array}\right] \left[ \begin{array}{c}
         \bl A^\top \\
         \bl N^\top
    \end{array}\right] (\bl X^\top \bl X)^{-1} \left[ \begin{array}{cc}
         \bl A & \bl N
    \end{array}\right] \bl Z^\top \bl y \\
    &= \bl A^\top (\bl X^\top \bl X)^{-1} \left[ \begin{array}{cc}
         \bl A & \bl N
    \end{array}\right] \left[ \begin{array}{c}
         \bl (\bl A^\top \bl A)^{-1} \bl A^\top \\
         \bl N^\top
    \end{array}\right] \bl X^\top \bl y \\
    &= \bl A^\top (\bl X^\top \bl X)^{-1} \bl X^\top \bl y \\ 
    \end{align*}
\end{proof}

\begin{proof}[Proof of Proposition \ref{prop:2}]
    Apply the \cite{seber_linear_2003} identity to find
    \begin{equation*}
        \bl I_p - \hat{\bl M} (\hat{\bl M}^\top \hat{\Sigma}_{\gamma}^{-1} \hat{\bl M})^{-1} \hat{\bl M}^\top \hat{\Sigma}_{\gamma}^{-1} = \hat{\Sigma}_{\gamma} \bl U (\bl U^\top \hat{\Sigma}_{\gamma} \bl U)^{-1} \bl U^\top,
    \end{equation*}
    where $\bl U$ is a $p \times (p - k)$ matrix whose columns form an orthonormal basis for the null space of $\hat{\bl M}^\top$. So
    \begin{equation*}
    \begin{aligned}
        \hat{\bs \theta}_{\gamma} &= (\hat{\bl M}^\top \hat{\bl M})^{-1} \hat{\bl M}^\top \bl y - (\hat{\bl M}^\top \hat{\bl M})^{-1} \hat{\bl M}^\top \hat{\Sigma}_{\gamma} \bl U (\bl U^\top \hat{\Sigma}_{\gamma} \bl U)^{-1} \bl U^\top \bl y \\
        &= (\hat{\bl M}^\top \hat{\bl M})^{-1} \hat{\bl M}^\top \bl y - (\hat{\bl M}^\top \hat{\bl M})^{-1} \hat{\bl M}^\top \bl E \bl E^\top \bl U (\bl U^\top \hat{\Sigma}_{\gamma} \bl U)^{-1} \bl U^\top \bl y
    \end{aligned}
    \end{equation*}
    As $\gamma \rightarrow \infty$, the second term goes to zero,
    so
    \begin{equation*}
        \lim_{\gamma \rightarrow \infty }~\hat{\bs \theta}_{\gamma} = \hat{\bs \theta}_{\mathrm{ATR}}.
    \end{equation*}
    Now, write $\bl U$ as
    \begin{equation*}
    \bl U = \left[ \begin{array}{cc}
    \bl U_1 & \bl U_2
    \end{array} \right]
    \end{equation*}
    where the $n-k$ columns of $\bl U_1$ are the left singular vectors of $\bl X (\bl X^\top \bl X)^{-1} \bl N$, and the $p - n$ columns of $\bl U_2$ are the left singular vectors of $(\bl I_p - \bl P_{\bl X})$. This is a valid choice for $\bl U$ because
    \begin{equation*}
    \begin{aligned}
	   \hat{\bl M}^\top \bl X (\bl X^\top \bl X)^{-1} \bl N &= \bl 0 \\
	   \hat{\bl M}^\top (\bl I_p - \bl P_{\bl X}) &= \bl 0,
    \end{aligned}
    \end{equation*}
    implying that the column spaces of $\bl U_1$ and $\bl U_2$ are both in the null space of $\hat{\bl M}^\top$, and $(\bl I_p - \bl P_{\bl X})\bl X (\bl X^\top \bl X)^{-1} \bl N = \bl 0$, which implies that $\bl U_2^\top \bl U_1 = \bl 0$. Let the singular value decomposition of $\bl X (\bl X^\top \bl X)^{-1} \bl N$ be $\bl U_1 \bl D_1 \bl V_1^\top$. A direct calculation shows that
    \begin{equation*}
    (\bl U^\top \hat{\Sigma}_{\gamma} \bl U)^{-1} = \left[\begin{array}{cc}
    \bl D_1^{-2} + \gamma \bl I_{n-k} & \bl 0 \\
    \bl 0 & \gamma \bl I_{p-n}
    \end{array} \right]^{-1}.
    \end{equation*}
    Also, $\bl E^\top \bl U = [\bl V_1 \bl D_1^{-1} ~ \bl 0]$. So
    \begin{equation*}
        \begin{aligned}
            \bl E^\top \bl U (\bl U^\top \hat{\Sigma}_{\gamma} \bl U)^{-1} \bl U^\top &= \bl V_1 \bl D_1^{-1} (\bl D_1^{-2} + \gamma \bl I_{n-k})^{-1} \bl U_1^\top \bl y.
        \end{aligned}
    \end{equation*}
    Hence, as $\gamma \rightarrow 0$, we have
    \begin{equation*}
        \begin{aligned}
            \lim_{\gamma \rightarrow 0}~ \hat{\bs \theta}_{\gamma} &= (\hat{\bl M}^\top \hat{\bl M})^{-1} \hat{\bl M}^\top \bl y - (\hat{\bl M}^\top \hat{\bl M})^{-1} \hat{\bl M}^\top \bl E \bl V_1 \bl D_1 \bl U_1^\top \bl y \\
            &= (\hat{\bl M}^\top \hat{\bl M})^{-1} \hat{\bl M}^\top \bl y - (\hat{\bl M}^\top \hat{\bl M})^{-1} \hat{\bl M}^\top \bl X (\bl I_n - \bl P_{\bl A}) (\bl X^\top \bl X)^{-1} \bl X^\top \bl y \\
            &= (\hat{\bl M}^\top \hat{\bl M})^{-1} \hat{\bl M}^\top (\bl I_p - \bl P_{\bl X}) \bl y + (\hat{\bl M}^\top \hat{\bl M})^{-1} \hat{\bl M}^\top \bl X \bl A (\bl A^\top \bl A)^{-1} \bl A^\top (\bl X^\top \bl X)^{-1} \bl X^\top \bl y \\
            &= \bl A^\top (\bl X^\top \bl X)^{-1} \bl X^\top \bl y \\
            &= \hat{\bs \theta}_{\mathrm{RTS}}
        \end{aligned}
    \end{equation*}
\end{proof}

\begin{proof}[Proof of Proposition 3]
	Some of the steps here are similar to those in the proof of Proposition \ref{prop:2}. From the definition of $\hat{\bl y}_{\gamma}'$, we have
	\begin{equation*}
    \hat{\bl y}_\gamma' = \hat{\bl M}' \hat{\bs \theta}_{0 \gamma} + \hat{\Delta}^\top \hat{\Sigma}_{0\gamma}^{-1}(\bl y_0 - \hat{\bl M}_0 \hat{\bs \theta}_{0 \gamma}).
\end{equation*}
The limiting behavior of the first term in the sum can be deduced directly from Proposition \ref{prop:2}:
\begin{equation*}
	\lim_{\gamma \rightarrow 0}~ \hat{\bl M}' \hat{\bs \theta}_{0 \gamma} = \hat{\bl M}' \bl A^\top (\bl X_0^\top \bl X_0)^{-1} \bl X_0^\top \bl y_0 
\end{equation*}
and
\begin{equation*}
	\lim_{\gamma \rightarrow \infty}~ \hat{\bl M}' \hat{\bs \theta}_{0 \gamma} = \hat{\bl M}' (\hat{\bl M}_0^\top \hat{\bl M}_0)^{-1} \hat{\bl M}_0^\top \bl y_0 
\end{equation*}
For the second term, observe that
\begin{equation*}
\begin{aligned}
	\bl y_0 - \hat{\bl M}_0 \hat{\bs \theta}_{0 \gamma} &= \bl y_0 - \hat{\bl M}_0 (\hat{\bl M}_0^\top \hat{\Sigma}_{0\gamma}^{-1} \hat{\bl M}_0)^{-1} \hat{\bl M}_0^\top \hat{\Sigma}_{0\gamma}^{-1} \bl y_0 \\
	&= (\bl I_p - \hat{\bl M}_0 (\hat{\bl M}_0^\top \hat{\Sigma}_{0\gamma}^{-1} \hat{\bl M}_0)^{-1} \hat{\bl M}_0^\top \hat{\Sigma}_{0\gamma}^{-1}) \bl y_0 \\
	&= \hat{\Sigma}_{0\gamma} \bl U (\bl U^\top \hat{\Sigma}_{0\gamma} \bl U)^{-1} \bl U^\top \bl y_0
\end{aligned}
\end{equation*}
where $\bl U$ is a $p \times (p - k)$ matrix whose columns form an orthonormal basis for the null space of $\hat{\bl M}_0^\top$. The last line is again a result of the \cite{seber_linear_2003} identity. Pre-multiplying by $\hat{\Delta}^\top \hat{\Sigma}_{0\gamma}^{-1}$, we see that the second term is equal to 
\begin{equation*}
\hat{\Delta}^\top \bl U (\bl U^\top \hat{\Sigma}_{0\gamma} \bl U)^{-1} \bl U^\top \bl y_0.
\end{equation*}

Now write $\bl U$ as
\begin{equation*}
\bl U = \left[ \begin{array}{cc}
\bl U_1 & \bl U_2
\end{array} \right]
\end{equation*}
where the $n-k$ columns of $\bl U_1$ are the left singular vectors of $\bl X_0 (\bl X_0^\top \bl X_0)^{-1} \bl N$, and the $p - n$ columns of $\bl U_2$ are the left singular vectors of $(\bl I_p - \bl P_{\bl X_0})$. This is a valid choice for $\bl U$ because
\begin{equation*}
\begin{aligned}
	\hat{\bl M}_0^\top \bl X_0 (\bl X_0^\top \bl X_0)^{-1} \bl N &= \bl 0 \\
	\hat{\bl M}_0^\top (\bl I_p - \bl P_{\bl X_0}) = \bl 0
\end{aligned}
\end{equation*}
implying that the column spaces of $\bl U_1$ and $\bl U_2$ are both in the null space of $\hat{\bl M}_0^\top$, and $(\bl I_p - \bl P_{\bl X_0})\bl X_0 (\bl X_0^\top \bl X_0)^{-1} \bl N = \bl 0$, which implies that $\bl U_2^\top \bl U_1 = \bl 0$. Let the singular value decomposition of $\bl X_0 (\bl X_0^\top \bl X_0)^{-1} \bl N$ be $\bl U_1 \bl D_1 \bl V_1^\top$. A direct calculation shows that
\begin{equation*}
(\bl U^\top \hat{\Sigma}_{0\gamma} \bl U)^{-1} = \left[\begin{array}{cc}
\bl D_1^{-2} + \gamma \bl I_{n-k} & \bl 0 \\
\bl 0 & \gamma \bl I_{p-n}
\end{array} \right]^{-1}.
\end{equation*}
Also, $\hat{\Delta}^\top \bl U = [\bl X' \bl N \bl V_1 \bl D_1^{-1} ~ \bl 0]$, so we have
\begin{equation*}
\begin{aligned}
\hat{\Delta}^\top \bl U (\bl U^\top \hat{\Sigma}_{0\gamma} \bl U)^{-1} \bl U^\top \bl y_0 &= \bl X' \bl N \bl V_1 \bl D_1^{-1} (\bl D_1^{-2} + \gamma \bl I_{n-k})^{-1} \bl D_1^{-1} \bl V_1^\top \bl N^\top (\bl X_0^\top \bl X_0)^{-1} \bl X_0^\top \bl y_0 .
\end{aligned}
\end{equation*}
The limit of the right hand side as $\gamma \rightarrow \infty$ is $\bl 0$. As $\gamma \rightarrow 0$ the limit of the right hand side is
\begin{equation*}
\begin{aligned}
\bl X' \bl N \bl N^\top (\bl X_0^\top \bl X_0)^{-1} \bl X_0^\top \bl y_0 .
\end{aligned}
\end{equation*}
Therefore,
\begin{equation*}
\lim_{\gamma \rightarrow \infty}~ \hat{\bl y}_\gamma' = \hat{\bl M}' (\hat{\bl M}_0^\top \hat{\bl M}_0)^{-1} \hat{\bl M}_0^\top \bl y_0.
\end{equation*}
and
\begin{equation*}
\begin{aligned}
\lim_{\gamma \rightarrow 0}~ \hat{\bl y}_\gamma' &= \hat{\bl M}' \bl A^\top (\bl X_0^\top \bl X_0)^{-1} \bl X_0^\top \bl y_0  + \bl X' \bl N \bl N^\top (\bl X_0^\top \bl X_0)^{-1} \bl X_0^\top \bl y_0 \\
&= \bl X' \bl P_{\bl A} (\bl X_0^\top \bl X_0)^{-1} \bl X_0^\top \bl y_0 + \bl X'(\bl I_n - \bl P_{\bl A})(\bl X_0^\top \bl X_0)^{-1} \bl X_0^\top \bl y_0 \\
&= \bl X' (\bl X_0^\top \bl X_0)^{-1} \bl X_0^\top \bl y_0
\end{aligned}
\end{equation*}
\end{proof}

\begin{proof}[Proof of Proposition \ref{prop:atrvsrts}]
    Based on \eqref{eq:vars}, it is clear that $\Var{\hat{\bs \theta}_{\mathrm{RTS}}} \preceq \Var{\hat{\bs \theta}_{\mathrm{ATR}}}$ if and only if
    \begin{equation*}
    \gamma \bl V_1 - \bl V_2 \preceq \bl 0.
    \end{equation*}
    Because $\bl V_1$ and $\bl V_2$ are both symmetric and positive definite, the simultaneous diagonalization lemma (see \cite{muirhead_aspects_2005} Theorem A9.9) guarantees the existence of a $K \times K$ invertible matrix $\bl F$ such that
\begin{equation*}
\begin{aligned}
\bl V_1 &= \bl F \bl F^\top \\
\bl V_2 &= \bl F \Lambda \bl F^\top,
\end{aligned}
\end{equation*}
where $\Lambda  = \mathrm{diag}(\lambda_1, \dots, \lambda_K)$ are the eigenvalues of $\bl V_1^{-1} \bl V_2$. This implies that $\gamma \bl V_1 - \bl V_2 \preceq \bl 0$ if and only if $\gamma - \lambda_k \leq 0$ for all $k \in \{1, \dots, K\}$. Hence, $\gamma \leq \lambda_{\min}(\bl V_1^{-1} \bl V_2)$ is necessary and sufficient for $\Var{\hat{\bs \theta}_{\mathrm{RTS}}} \preceq \Var{\hat{\bs \theta}_{\mathrm{ATR}}}$.
\end{proof}

\begin{proof}[Proof of Proposition \ref{prop:sebias}]
    Define the orthogonal matrices $\bl U_1, \bl U_2$ and $\bl U_3$ as in the statement of the proposition. The variance of the $k$th entry of the RTS estimates is
    \begin{equation*}
    v_k = \|\bs \theta\|_2^2 \bl a_k^\top (\bl X^\top \bl X)^{-1} \bl X^\top \Sigma \bl X (\bl X^\top \bl X)^{-1} \bl a_k .
    \end{equation*}
    Let $\hat{v}_k$ be the squared standard error of the RTS estimate proposed in \eqref{eq:sse}. Using properties of the trace, we can write
    \begin{equation*}
    \begin{aligned}
    \mathrm{E}[\hat{v}_k] &= \mathrm{E}[\tr((\bl I_p - \bl P_{\bl X}) \bl y \bl y^\top)]/(p - n) \\
    &= \mathrm{E}[\tr(\bl U_3 \bl U_3^\top \bl y \bl y^\top)]/(p - n) \\
    &= [\tr(\bl U_3 \bl U_3^\top (\|\bs \theta\|_2^2 \Sigma + \bl M \bs \theta \bs \theta^\top \bl M^\top))]/(p - n) \\
    &=\bl a_k^\top (\bl X^\top \bl X)^{-1} \bl a_k [\|\bs \theta\|_2^2 \tr(\bl U_3^\top \Sigma \bl U_3) + \tr(\bl U_3^\top \bl M \bs \theta \bs \theta^\top \bl M^\top \bl U_3)]/(p-n).
    \end{aligned}
    \end{equation*}
    To derive the upper bound in the proposition, see that
    \begin{equation*}
    \begin{aligned}
    v_k - \mathrm{E}[\hat{v}_k] &= \|\bs \theta\|_2^2 \bl a_k^\top (\bl X^\top \bl X)^{-1} \bl X^\top \Sigma \bl X (\bl X^\top \bl X)^{-1} \bl a_k - \\
    &~~~~\bl a_k^\top (\bl X^\top \bl X)^{-1} \bl a_k [\|\bs \theta\|_2^2 \tr(\bl U_3^\top \Sigma \bl U_3) + \tr(\bl U_3^\top \bl M \bs \theta \bs \theta^\top \bl M^\top \bl U_3)]/(p-n) \\
    &\leq \|\bs \theta\|_2^2 \bl a_k^\top (\bl X^\top \bl X)^{-1} \bl X^\top \bl U_2 \bl U_2^\top \Sigma \bl U_2 \bl U_2^\top \bl X (\bl X^\top \bl X)^{-1} \bl a_k - \\
    &~~~~\|\bs \theta\|_2^2 \bl a_k^\top (\bl X^\top \bl X)^{-1} \bl a_k \tr(\bl U_3^\top \Sigma \bl U_3)/(p-n) \\
    &\leq \|\bs \theta\|_2^2 \bl a_k^\top (\bl X^\top \bl X)^{-1} \bl a_k \lambda_{\max}(\bl U_2^\top \Sigma \bl U_2 )- \\
    &~~~~\|\bs \theta\|_2^2 \bl a_k^\top (\bl X^\top \bl X)^{-1} \bl a_k \tr(\bl U_3^\top \Sigma \bl U_3)/(p-n) \\
    &= \|\bs \theta\|_2^2 \bl a_k^\top (\bl X^\top \bl X)^{-1} \bl a_k [\lambda_{\max}(\bl U_2^\top \Sigma \bl U_2 )- \bar{\lambda}(\bl U_3^\top \Sigma \bl U_3) ] .\\
    \end{aligned}
    \end{equation*}
    To derive the lower bound,
    \begin{equation*}
    \begin{aligned}
    v_k - \mathrm{E}[\hat{v}_k] &= \|\bs \theta\|_2^2 \bl a_k^\top (\bl X^\top \bl X)^{-1} \bl X^\top \Sigma \bl X (\bl X^\top \bl X)^{-1} \bl a_k - \\
    &~~~~\bl a_k^\top (\bl X^\top \bl X)^{-1} \bl a_k [\|\bs \theta\|_2^2 \tr(\bl U_3^\top \Sigma \bl U_3) + \tr(\bl U_3^\top \bl M \bs \theta \bs \theta^\top \bl M^\top \bl U_3)]/(p-n) \\
    &\geq \|\bs \theta\|_2^2 \bl a_k^\top (\bl X^\top \bl X)^{-1} \bl X^\top \bl U_2 \bl U_2^\top \Sigma \bl U_2 \bl U_2^\top \bl X (\bl X^\top \bl X)^{-1} \bl a_k - \\
    &~~~~\|\bs \theta\|_2^2 \bl a_k^\top (\bl X^\top \bl X)^{-1} \bl a_k [\tr(\bl U_3^\top \Sigma \bl U_3) + \tr(\bl U_3^\top \bl M \bl M^\top \bl U_3)]/(p-n) \\
    &\geq \|\bs \theta\|_2^2 \bl a_k^\top (\bl X^\top \bl X)^{-1} \bl a_k [\lambda_{\min}(\bl U_2^\top \Sigma \bl U_2 )- \bar{\lambda}(\bl U_3^\top \Sigma \bl U_3) - \bar{\lambda}(\bl U_3^\top \bl M \bl M^\top \bl U_3) ] .\\
    \end{aligned}
    \end{equation*}
\end{proof}

\end{document}